\documentclass[runningheads]{llncs}

\usepackage[utf8]{inputenc}

\usepackage[T1]{fontenc}    
\usepackage{booktabs}       

\usepackage{amsfonts,amsmath,amssymb,bbm}
\usepackage{enumitem}
\usepackage[table]{xcolor}
    \usepackage{adjustbox}
\usepackage[UKenglish]{isodate}
\usepackage{hyperref} 
    \hypersetup{
    	colorlinks = true,
    	linkcolor = blue,
    	anchorcolor = blue,
    	citecolor = blue,
    	filecolor = blue,
    	urlcolor = blue
    }
\usepackage{float}
\usepackage{enumitem}
\usepackage{multirow}
\usepackage{fancyvrb}
\usepackage{rotating}

\renewcommand{\phi}{\varphi}

\newcommand{\rr}{\mathbb{R}}

\newcommand{\eqdef}{\ensuremath{\stackrel{\mbox{\upshape\tiny def.}}{=}}}

\usepackage{appendix}

\usepackage{booktabs}
\usepackage{lscape}
\usepackage{graphicx}

\date{February 20, 2023}

\begin{document}
\title{Generative Ornstein–Uhlenbeck Markets via Geometric Deep Learning} %
\author{Anastasis Kratsios\inst{1}\orcidID{0000-0001-6791-3371} 
\and \\ Cody Hyndman\inst{2,3}\orcidID{0000-0001-5849-1584}
}

\authorrunning{A. Kratsios and C. Hyndman}

\institute{Department of Mathematics and Statistics, McMaster University, Canada 
\email{kratsioa@mcmaster.ca}\\
\url{https://anastasiskratsios.github.io/} 
\and
Department of Mathematics and Statistics, Concordia University, Canada\\
\email{cody.hyndman@concordia.ca}\\
\url{http://mypage.concordia.ca/alcor/chyndman/}
}
\maketitle             
\begin{abstract}
We consider the problem of simultaneously approximating the conditional distribution of market prices and their log returns with a single machine learning model.  We show that an instance of the GDN model of \cite{kratsios2022universal} solves this problem without having prior assumptions on the market's ``clipped'' log returns, other than that they follow a generalized Ornstein-Uhlenbeck process with a priori unknown dynamics.  We provide universal approximation guarantees for these conditional distributions and contingent claims with a Lipschitz payoff function. 

\end{abstract}

\noindent


\section{Introduction}

In classical portfolio theory, one considers a portfolio comprised of $D$ predetermined risky assets and a riskless asset.  The objective is to identify the most ``efficient portfolios'' by which we mean portfolios exhibiting the greatest gains while not exceeding a fixed level of risk or variability.  Here, a portfolio's gains are quantified by its expected (log) returns, and its risk is quantified by the variance of its (log) returns.  Thus, efficient portfolios are defined by optimizers of the following problem
\begin{equation}
\label{eq:SharpRatio}
\hat{w}(\gamma,\mu,\Sigma)\triangleq 
\underset{
	\underset{
		\bar{1}^{\star}w=1
	}{w \in \rr^D}
}{
	\operatorname{argmin}}
\left(
-\gamma\mu^{\star}w +\frac{
	w^{\star}\Sigma w
}{2}
\right)
.
\end{equation}
In~\eqref{eq:SharpRatio}, $w$ is the vector of portfolio weights expressed as the proportion of wealth invested in each risky asset, $\mu \in \rr$ is the vector of the expected (log) returns of the risky assets, $\Sigma$ is the covariance matrix of that portfolio's (log) returns, $\gamma$ is a parameter balancing the objectives of maximizing the portfolio return versus minimizing the portfolio variance, $\bar{1}$ is the vector with all its components equal to $1$, and $^\star$ denotes matrix transpose operator.  
If $\Sigma$ is non-singular, the unique optimal solution to equation~\eqref{eq:SharpRatio} is given in closed-form by
\begin{equation}
\label{eq:optim_portfolio}
\hat{w}(\gamma,\mu,\Sigma) = 
\frac{\Sigma^{-1}\bar{1}}{\bar{1}^{\star} \Sigma^{-1} \bar{1}}
+ \gamma \left( \Sigma^{-1}\mu - \frac{\bar{1}^{\star}\Sigma^{-1}\mu}{\bar{1}^{\star}\Sigma^{-1}\bar{1}}
\Sigma^{-1}\bar{1}
\right)
.
\end{equation}
The particular case where $\gamma$ is set to $0$ is the minimum variance portfolio of \cite{markowitz1968portfolio}.  The minimum-variance portfolio $\hat{w}(0,\mu,\Sigma)$ may also be derived by minimizing the portfolio variance subject to the budget constraint $\bar{1}^\star w=1$.  Accordingly, we consider the case where $\Sigma$ is \textit{non-singular}.  The optimality of~\eqref{eq:SharpRatio} is contingent on the \textit{normality} of the asset's (log) returns in this \textit{static} picture of the market.

In reality, any financial market is continually and randomly evolving.  Therefore, one must actively update the risky asset's mean $\mu$ and covariance matrix $\Sigma$ in~\eqref{eq:optim_portfolio} to maintain an efficient portfolio.  Since future market prices are unknown, so are the efficient portfolios in~\eqref{eq:SharpRatio}.  Thus, our \textit{objective} will be to forecast \textit{both} the conditional evolution of market prices and the distribution of their log returns up to a regularizing factor.

\paragraph{Encoding Market Dynamics via Clipped log returns}
In stochastic finance, the market's continual random evolution is typically formalized by a $(0,\infty)^D$-valued stochastic process $S_{\cdot}\eqdef (S_t)_{t\ge 0}$ defined on a complete filtered probability space $(\Omega,\mathcal{F},\mathbb{F}\eqdef (\mathcal{F}_t)_{t\ge 0},\mathcal{P})$.  The components of $S_{\cdot}$ describe the evolving market prices.  For simplicity, we omit the riskless asset, assuming that the continuously compounded interest rate is a constant, $r\ge 0$.

Since problem~\eqref{eq:SharpRatio}, concerns the log returns of the market's assets, i.e.\ one often models a latent Gaussian process $X_{\cdot}$ driving the market prices where $S_t\approx e^{X_t}$ (where the exponential map is applied component-wise).  
This is primarily due to three reasons: 1) stock prices cannot be non-positive, 2) most stock returns are somewhat log-normally distributed on an appropriate time-scale, and 3) the distribution of a stock's log returns are mathematically convenient.  

We note that any asset's (log) returns can be substantial, either in the negative or positive directions, but realistically they cannot be arbitrarily large.  With this in mind, it will be analytically convenient to work with ``clipped (or regularized) log returns'' which also satisfy the heuristics (1)-(3).  By ``clipped log returns'' we encode the evolution of the market's prices $S_{\cdot}$ as
\begin{equation}
\label{eq:clipped_price_process}
\begin{aligned}
S_t 
\eqdef  
\mathcal{E}(X_t) 
\mbox{ and }
        \mathcal{E}(x)
    \eqdef 
        \exp\biggl(
            \frac{1}{\min\{1,\|x/M\|\}}\cdot x
        \biggr)
\end{aligned}
\end{equation}
for all $t\ge 0$, where the exponential map $\exp$ is \textit{applied component-wise} to any vector in $\mathbb{R}^D$, and the ``clipping threshhold'' $M>0$ is a fixed and large.  From a practical standpoint, both ways of encoding the evolution of market prices $e^{X_t}$ and $\mathcal{E}(X_t)$, into the latent ``log returns-like'' Gaussian process $X_{\cdot}$, are virtually indistinguishable for $M$ large enough.  The main technical advantage of $\mathcal{E}$ over $\exp$ is that it Lipschitz; thus, it is compatible with the optimal-transport toolbox.  

The transformation $\mathcal{E}$ is also appealing from the \textit{stochastic analytic} vantage point.  This is because it is the composition of a convex function with a smooth function; whence, if $X_{\cdot}$ is a semi-martingale, then we can directly compute the dynamics of $S_{\cdot}$ from those of $X_{\cdot}$ using a non-smooth It\^{o} formula (e.g.\ \cite{ProtterCarlen_1992_ItoConvexNonsmooth_IJM,FollmerProtter__NonsmoothIto_PTRF2000}).

\paragraph{An Interpretable but Model-Agnostic Approach}
We operate in the interpretable scenario where the clipped log returns process $X_{\cdot}$'s are not only conditionally Gaussian, but they are a strong solution to a simple and interpretable \textit{stochastic differential equation (SDE)}.  We consider the generalized Ornstein-Uhlenbeck (OU) process
\begin{equation}
\label{eq:SDE__assumed}
X_t^x = x + \int_0^t \, (\mu_s + M_s\,X_s^x)\,ds + \int_0^t\, \sigma_s\,dW_s,
\end{equation}
where $W_{\cdot}\eqdef (W_t)_{t\ge 0}$ is a $D$-dimensional $\mathbb{F}$-Brownian motion, $\alpha:\mathbb{R}\rightarrow \mathbb{R}^D$ and $\beta:\mathbb{R}\rightarrow \mathbb{R}^{D\times D}$ are a-priori \textit{unknown} continuously differentiable Lipschitz functions, and $\sigma:\mathbb{R}\rightarrow \mathbb{R}^{D\times D}$ is an a-priori \textit{unknown} Lipschitz functions; further each $\sigma_t$ a symmetric positive definite matrix (for $t\ge 0$).  
We drop the superscript emphasizing the dependence of $X_{\cdot}^x$ on the initial condition $x$ whenever clear from the context.

The first appeal of~\eqref{eq:SDE__assumed} is that, given any $\mu_{\cdot}$ and any $\sigma_{\cdot}$, the dynamics of $X_t$ and $S_t\approx X_t$ are readily \textit{interpretable}. 
The second appeal of~\eqref{eq:SDE__assumed}, after a simple/classical computation, shows that each $X_t$ follows a $D$-dimensional Gaussian distribution with mean $\int_0^t\,\mu_s \,ds$ and \textit{non-singular} covariance $\int_0^t\, \sigma_s \sigma_s^{\top}\,ds$; which we denote $\mathcal{N}_D\big(\bar{\mu}_t,\,\int_0^t\, \sigma_s\sigma_s^{\top}\,ds\big)$ where $\bar{\mu}_t$ solves the ODE $\partial_t{\bar{\mu}}_t = \mu_t + M_t\bar{\mu}_t$ for the initial condition $\mu_0=x$.  
Note that if $M_t=0$ then $\bar{\mu}_t=\int_0^t\, \mu_s\,ds$.

As an informal illustration, suppose that $\mu_t = \mu_0 - \sigma_0^2/2$, $M_t=0$, and $\sigma_t=\sigma_0$ in~\eqref{eq:SDE__assumed}.  Then, as $M$ tends to infinity, we see $S_t$ tends to the classical Geometric Brownian Motion (GBM) model used to derive the classical Black-Scholes formula and used to derive tractable optimal investment strategies \cite{YuChingWeChufandGu__JOTA_2023,GatheralSchied_GBMTradinig_IJTAF_2011}.

We remain \textit{agnostic} to specifications of $\mu$ and of $\sigma$ and instead, we adopt a machine learning approach.  Our first main objective is to implicitly infer the dynamics of $X_{\cdot}$by explicitly approximating its \textit{regular conditional distribution function} $x\mapsto \mathbb{P}(X_t\in \cdot\vert X_0=x)$.  Then, our second goal is to deduce the same for $S_{\cdot}$.  
Thus, we instead only postulate minimal regularity of the functions $\mu_{\cdot}$ and $\sigma_{\cdot}$, just enough for a \textit{deep neural network approximation} to the conditional probability distribution function of $X_{\cdot}$  to be viable.  

\paragraph{Contributions}
We will show that the \textit{geometric deep network} modelling framework of \cite{kratsios2022universal}, as specified in \cite[Corollary 39]{kratsios2022universal}, provides a universal solution to the problem of simultaneously predicting the regular conditional distributions of $X_{\cdot}$ and of $S_{\cdot}$, conditioned on the current state of the market $x$ for any given future time $t$.
In this case, the GDN implements a principled extension of the so-called \textit{deep Kalman filter} of \cite{krishnan2015deep}, which has recently also entered the mathematical finance literature in \cite{jaimungal2022reinforcement}.

\vspace{-.5em}
\subsection*{Relation to Other Deep Learning Models}

There have recently been several other probability-measure-valued deep learning models proposed in the literature.  For instance, \cite{Acciaio2022_GHT} proposes a deep learning framework for approximating any regular conditional distribution function when the target space of probability measures is equipped with the $1$-Wasserstein or adapted $p$-Wasserstein distances.  In the case of the simple market dynamics~\eqref{eq:SDE__assumed}, we will find that the GDN model is more economical in its theoretically guaranteed parameter count.  Unlike those models, its approximation-theoretic guarantees are necessarily limited to markets evolving according to generalized OU dynamics such as~\eqref{eq:SDE__assumed}, with non-singular volatility/diffusion.  
Gaussian-measure-valued deep learning models were experimentally considered in \cite{krishnan2015deep}.  


Additional results can be found in the arXiv version, while experimental support is provided at \cite{anastasisGit}.  

\section{Preliminaries}
\label{s:Preliminaries}
We review the necessary background required to formulate our main results.

\vspace{-1em}
\subsection{$2$-Wasserstein Riemannian Geometry}
\label{s:Preliminaries__WassersteinRiemannian}
We equip the set of $D$-dimensional Gaussian distributions with non-singular covariance, denoted by $\mathcal{N}_D$, with a smooth structure induced by the global chart
\begin{equation}
\label{eq:GlobalChart_ND}
\tag{Chart}
\begin{aligned}
\phi:\mathbb{R}^D\times \mathbb{R}^{D(D+1)/2} & \rightarrow \mathcal{D}\\
(\mu,\sigma) & \mapsto \mathcal{N}_D\Big(
    \mu,
    \exp\circ \operatorname{sym}(\sigma)
\Big)
,
\end{aligned}
\end{equation}
where $\exp$ is the matrix exponential and $\operatorname{sym}$ is the linear map sending any vector $X \in \mathbb{R}^{D(D+1)/2}$ to $D\times D$ symmetric matrix
\begin{equation}
\label{eq:Sym_Function}
\resizebox{.45 \textwidth}{!}{$
        \operatorname{sym}(X)
    \eqdef 
        \begin{pmatrix}
        X_{1}  & X_{2} &  \ldots  & X_{D} \\
        X_{2} & X_{D+1} &  \ldots  & X_{2D-1} \\
        \vdots&  & \ddots & \vdots \\
        X_D &  & \ldots      & X_{D(D+1)/2}.
        \end{pmatrix} 
$}
\end{equation}
Following \cite{takatsu2011wasserstein}, we equip $\mathcal{N}_D$ with a Riemannian
metric $g_{\mathcal{W}_2}$ defined at any $D$-dimensional Gaussian distribution $\mathcal{N}_{D}(\mu,\Sigma)$ with non-singular covariance matrix (i.e.\ any point in $\mathcal{N}_D)$) by
\[
    g_{\mathcal{W}_2,\,(\mu,\Sigma)}(u,v)
    \eqdef 
    \langle u_1,v_1\rangle 
    +
    \operatorname{tr}\big(\operatorname{sym}(u_2)\Sigma \operatorname{sym}(u_2)\big)
,
\]
where we have identified the tangent vectors $u,v$ at $\mathcal{N}_D(\mu,\Sigma)$ with Euclidean vectors via $u=(u_1,u_2),v=(v_1,v_2)\in \mathbb{R}^{D}\times \mathbb{R}^{D(D+1)/2}$.  Together $(\mathcal{N}_D,g_{\mathcal{W}_2})$ is a well-defined simply connected Riemannian manifold (whence it has a well-defined geodesic distance between any two points).  In \cite[Proposition A]{takatsu2011wasserstein}, the authors show that the geodesic distance on $(\mathcal{N}_D,g_{\mathcal{W}_2})$ coincides with the $2$-Wasserstein distance $\mathcal{W}_2$ on $\mathcal{N}_D$.  By \cite{gelbrich1990formula}, $\mathcal{W}_2$ admits the following closed-form for any $\mathcal{N}_D(\mu_1,\Sigma_1),\mathcal{N}_D(\mu_2,\Sigma_2)\in \mathcal{N}_D$
\[
       \mathcal{W}_2^2\big(
           \mathcal{N}_D(\mu_1,\Sigma_1)
       ,
           \mathcal{N}_D(\mu_2,\Sigma_2)
       \big) 
   = 
       \| \mu_1 - \mu_2 \|^2 + \operatorname{tr} \big( \Sigma_1 + \Sigma_2 - 2 ( \Sigma_2^{1/2} \Sigma_1 \Sigma_2^{1/2} )^{1/2} \bigr)
   ,
\]
where $\Sigma_i^{1/2}$ denotes the square-root of the positive-definite matrices $\Sigma_1$ and $\Sigma_2$. 

\subsection{The GDN Model}
\label{s:GDNs}

Figure~\ref{fig_finite_dimensional_congugation_approach} illustrates the GDN implements the top arrow between the (non-Euclidean) Riemannian manifolds $(\mathcal{N}_D,g_{\mathcal{W}_2})$ in two phases.  First, it transforms that vector via a deep feedforward neural network with ReLU activation function; then, it decodes the deep feedforward neural network output by interpreting them as the parameters defining $D$-dimensional Gaussian distribution with a non-degenerate covariance matrix, thus generating a $\mathcal{N}_D$-valued prediction.  

\begin{figure}[ht!]
\centering
\includegraphics[width=1\linewidth]{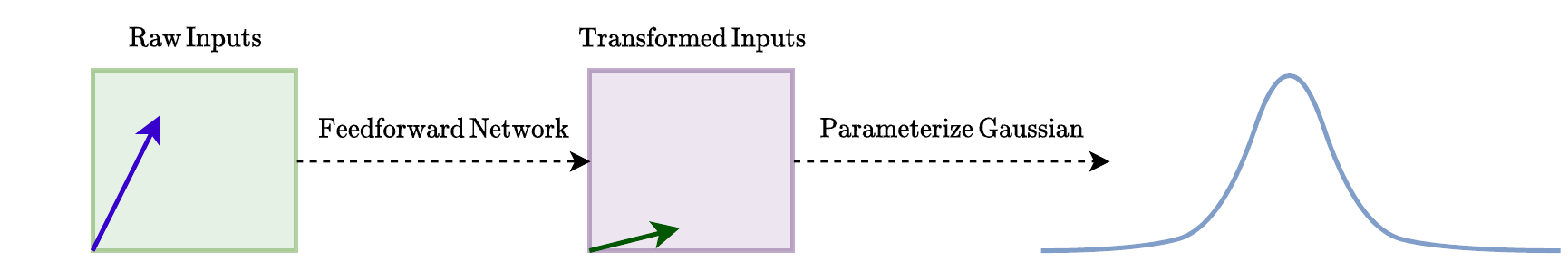}
\caption{\textit{Summary of the GDN Model Processing:}  First, it transforms the vectorial data using a deep feedforward neural network with a suitable activation function; next, the neural network output vectors are decoded as the parameters defining a Gaussian mean and covariance.  This Gaussian distribution is the GDN generated prediction.}
\label{fig_finite_dimensional_congugation_approach}
\end{figure}

\vspace{-1em}
\begin{definition}[Geometric Deep Network (GDN)]
\label{defn:GDN}
Fix a non-polynomial smooth ``activation function'' $\sigma:\mathbb{R}\rightarrow \mathbb{R}$, and $D,W,J\in \mathbb{N}_+$, a \textit{geometric deep network (GDN)} on $\mathcal{N}_D$ of \textit{width} $W$ and \textit{depth} $J$ is a map $\hat{f}:\mathbb{R}^{1+d}\rightarrow \mathcal{N}_D$ with representation: for every $\mathcal{N}_D(\mu,\Sigma)\in \mathcal{N}_D$
\[
\begin{aligned}
    \hat{f}(x) & \eqdef \varphi(A^{(J)}\,x^{(d)} + b^{(J)})\\
    x^{(k+1)} & \eqdef \sigma\bullet (A^{(k)}\,x^{(k)} + b^{(k)}) \qquad \mbox{for } k=0,\dots,J-1\\
    x^{(0)} & \eqdef (x,t),
\end{aligned}
\]
where $(x,t)\in \mathbb{R}^{1+d}\cong \mathbb{R}^d\times \mathbb{R}$, each $A^{(k)}\in \mathbb{R}^{d_k\times d_{k+1}}$, $b^{(k)}\in \mathbb{R}^{d_{k+1}}$, $D=d_0$, $d_J=D(D+1)/2$, $d_k\le W$ for every $k=1,\dots,J-1$, and $\phi$ as in~\eqref{eq:GlobalChart_ND}.
\end{definition}

\section{Main Results}
\label{s:Main_Results}
Our first result guarantees that the GDN model can approximate the distribution of $X_t^x$ at any future time $t$, given any initial state $x$, log returns imposing any modelling assumptions for the ``asset's drift'' $\mu_{\cdot}$ nor for its ``volatility'' $\sigma_{\cdot}$.  
\begin{theorem}[GDNs can Approximately Implement the Distribution of~\eqref{eq:SDE__assumed}]
\label{thrm:Main}
Fix a parameter $\delta>0$.
Let $K\subseteq \mathbb{R}^D$ be (non-empty) and compact and consider a ``time-horizon'' $T>\delta>0$.  For every ``approximation error'' $\epsilon>0$, there is a GDN $\hat{f}:\mathbb{R}^{1+D}\rightarrow \mathcal{N}_D$ satisfying the uniform estimate
\[
    \max_{x\in K,\,\delta\le t\le T}
        \mathcal{W}_2\big(
            \hat{f}(x,t)
        ,
            \mathbb{P}(X_t^x\in \cdot)
        \big)
    <
        \epsilon
.
\]
Moreover, if $t$ is fixed then $\hat{f}$ has width $D(6+2D+D^2)/2$ and depth $\mathcal{O}\big(\frac1{\epsilon^{2D}}\big)$.
\end{theorem}

The power of the GDN model is that it can \textit{simultaneously approximate} the regular conditional distribution (RCD) of the clipped log returns process and the market prices.  

\begin{theorem}[Simultaneous Approximation of the Market RCD]
\label{thrm:simultaneous_approximation__S}
Consider the setting of Theorem~\ref{thrm:Main}, and let $\hat{f}$ the GDN obtained from that result.  For every $x\in K$ and each $t\in [\delta,T]$
\[
        \mathcal{W}_2\big(
            \mathbb{P}(S_t^x\in \cdot)
        ,   
            \mathcal{E}_{\#}(\hat{f}(t,x))
        \big)
    < 
        \sqrt{D}e^M
        \,
        \varepsilon
.
\]
\end{theorem}
The Fundamental Theorem of Asset Pricing \cite{delbaen1994general}, implies that risk-neutral prices for contingent claims exist and are expressed as conditional expectations of the claim payoffs, computed under an equivalent martingale measure for the discounted market prices $(e^{-rt}S_t^x)_{t\ge 0}$.  
For illustrative simplicity, suppose that $r=0$ and that $S_{\cdot}^x$ is a $\mathbb{P}$-martingale.  Then contingent claims are computed as conditional expectations under $\mathbb{P}$.

\begin{theorem}[Automatic Contingent Claim Pricing]
\label{thrm:Main__Pricing}
Consider the setting and conclusion of Theorem~\ref{thrm:simultaneous_approximation__S}, let $r=0$, and suppose that $S_{\cdot}$ is a $\mathbb{P}$-martingale.  
For every \textit{Lipschitz} payoff function $V:\mathbb{R}^D\rightarrow \mathbb{R}$, and every $(x,t)\in K\times [\delta,T]$
\[
    \Big\vert
        \mathbb{E}_{U\sim \hat{f}(x,t)}\big[V\big(\mathcal{E}(U)\big)\big]
        -
        \mathbb{E}_{\mathbb{Q}}\big[V({S}_t^x)\big]
    \Big\vert
    <
        C\epsilon
,
\]
for some constant $C\ge 0$ depending only on $V$, $D$ and on $M$.
\end{theorem}
Theorem~\ref{thrm:Main__Pricing} implies that once $\hat{f}$ is trained, then we can directly approximate any contingent claim on $S_{\cdot}$ by simply sampling $V(\mathcal{E}(U))$ where $U$ is distributed according to $\hat{f}(x,t)$.  
We conclude by proving our guarantees for the GDN model.  
\vspace{-1em}
\section{Proofs}
\begin{lemma}[Gaussianity of the SDE~\eqref{eq:SDE__assumed} Solutions]
\label{lem:Gaussianity_SDE}
For any $x\in \mathbb{R}^D$ and any $t>0$, the random vector $S_t^x$ is distributed according to a $D$-dimensional Gaussian distribution with non-singular covariance; more precisely
\[
S_t^x \sim 
\mathcal{N}_D\Big(
\bar{\mu}_t
,
\int_0^t\, \sigma_s\sigma_s^{\top}
\,dt
\Big);
\]
where $\bar{\mu}$ is continuous and solves $\partial_t{\bar{\mu}}_t = \mu_t + M_t\bar{\mu}_t$ with initial condition $\bar{\mu}_0=x$.
\end{lemma}
Lemma~\ref{lem:Gaussianity_SDE} implies that the map $(x,t)\mapsto \mathbb{P}(S_t^x\in \cdot)$ takes values in $(\mathcal{N}_D,g_{\mathcal{W}_1})$ so that we can apply the universal approximation theorem of \cite{kratsios2022universal}.  However, need to verify that the target function is continuous.  The next Lemma implies that $(x,t)\mapsto \mathbb{P}(S_t^x\in \cdot)$ has the required regularity to apply the results of \cite{kratsios2022universal}.
\begin{lemma}[Stability Estimate for $(t,x)\mapsto X_t^x$]
\label{lem:Stability}   
Fix a compact subset $K\subseteq \mathbb{R}^D$ and a positive ``time--horizon'' $T>0$.  Then the map 
\[
\begin{aligned}
\mathbb{R}^D\times [0,\infty) & \rightarrow (\mathcal{N}_D,g_{\mathcal{W}_2})\\
(x,t) & \mapsto X_t^x
\end{aligned}
\]
is Lipschitz in $x$ and $1/2$-H\"{o}lder in $t$, over $K\times [0,T]$.
\end{lemma}
\begin{proof}
For each $t,s>0$ and every $x,\tilde{x}\in \mathbb{R}^D$, we have
\begin{equation}
\label{PROOF_lem:Gaussianity_SDE__A}
\mathcal{W}_1\big(
\mathbb{P}(X_t^x\in \cdot),\mathbb{P}(S_s^{\tilde{x}}\in \cdot)
\big)
\le 
\mathbb{E}\big[
\|X_t^x - X_s^{\tilde{x}}\|^2
\big]^{1/2}
.
\end{equation}
Applying \cite[Propositions 8.15 and 8.16]{da2014introduction} to the right-hand side of~\eqref{PROOF_lem:Gaussianity_SDE__A} yields
\begin{equation}
\label{PROOF_lem:Gaussianity_SDE__B}
\mathcal{W}_2\big(
\mathbb{P}(X_t^x\in \cdot),\mathbb{P}(X_s^{\tilde{x}}\in \cdot)
\big)
\le 
\mathbb{E}\big[
\|X_t^x - X_s^{\tilde{x}}\|^2
\big]^{1/2}
\le 
C(\|t-s\|^{1/2}+\|x-\tilde{x}\|)
,
\end{equation}
for some constant $C\ge 0$ depending on $K\times [0,T]$.  The result then follows from \cite[Proposition A]{takatsu2011wasserstein}, which states that the geodesic distance on $(\mathcal{N}_D,g_{\mathcal{W}_2})$ coincides with the restriction of the $2$-Wasserstein distance thereto.  
\end{proof}

\label{s:Proofs}
\begin{proof}[{Proof of Theorem~\ref{thrm:Main}}]
By Lemma~\ref{lem:Gaussianity_SDE} for every $(x,t)\in K \times [\delta,T]$ the random vector $X_t^x$ is distributed according to a Gaussian distribution with a non-singular covariance matrix.  Thus, the map $f(x,t)\mapsto \mathbb{P}(X_t^x\in \cdot)$ takes values in $\mathcal{N}_D$.  By our stability estimate, namely Lemma~\ref{lem:Stability}, $f$ is a Lipschitz function; in particular, it is continuous.  Therefore, \cite[Corollary 40]{kratsios2022universal} applies; whence, for every given $\epsilon>0$ there is a GDN satisfying $\max_{(x,t)\in K\times [\delta,T]}\,\mathcal{W}_2\big(f(x,t),\hat{f}(x,t)\big)<\epsilon$.
Furthermore, if $t$ is fixed, then the depth and width of $\hat{f}$ are given in the first row of~\cite[Table 1]{kratsios2022universal}; since $x\mapsto \mathbb{P}(X_t^x)$ is Lipschitz.
\end{proof}

\begin{lemma}
\label{lem:regularity_clippedexponential}
The push-forward $\mathcal{E}_{\#}$ is a well-defined map from $(\mathcal{P}_2(\mathbb{R}^D),\mathcal{W}_2)$ of $\sqrt{D}e^M$-Lipschitz continuity.  In particular, $\mathcal{E}_{\#}$ is a Lipschitz map to $(\mathcal{P}_2(\mathbb{R}^D),\mathcal{W}_1)$.
\end{lemma}
\begin{proof}[{Proof of Lemma~\ref{lem:regularity_clippedexponential}}]
We first observe that $\mathcal{E}$ is Lipschitz.  To see this, note that $x\mapsto \exp\big((\min\{1,\|x/M\|\})^{-1}\cdot x\big)$ is precisely the orthogonal projection $P$ of $\mathbb{R}^d$ onto the closed Euclidean ball $\overline{B}_{\mathbb{R}^D,\|\cdot\|_2}(0,M)$ of radius $M>0$ of about $0\in \mathbb{R}^D$.  Since $\overline{B}_{\mathbb{R}^D,\|\cdot\|_2}(0,M)$ is a closed convex set, then this projection is well-defined and $1$-Lipschitz (see \cite[Example 12.25 and Proposition 12.27]{BauschkeCombettes_2017_ConvexAnalysisHilberMonotoneCMSBook}).  Since $\mathcal{E}$ is given by the composition $\mathcal{E}=\exp\circ P$ (here $\exp$ is ``composed'' component-wise), $P$ is $1$-Lipschitz, and since the composition of Lipschitz functions is again Lipschitz, then $\mathcal{E}$ if $\exp$ is Lipschitz on the range of $P$.  
By Rademacher's theorem (see \cite[Theorem 3.16]{federer2014geometric}), if $\exp$ were to be $L$-Lipschitz on the range of $P$, then $\sup_{x\in P(\mathbb{R}^D)}\, \|\nabla \exp(x)\|$ must be finite; in which case this quantity is equal to its Lipschitz constant $L$.  This is indeed the case, since
\begin{equation}
\label{PROOF_eq:lem:regularity_clippedexponential___Rademacher}
    \sup_{x\in p(\mathbb{R}^D)}\, \|\nabla \exp(u)\|\le \sqrt{D}\, \max_{-M\le v\le M}\, \sqrt{D} e^M < \infty.
\end{equation}
Thus, $\mathcal{E}$ is -Lipschitz, with Lipschitz constant bounded-above by $L\eqdef \sqrt{D}e^M$.  

It is straight-forward to see that $\mathcal{E}_{\#}$ is well-defined and maps $\mathcal{P}_2(\mathbb{R}^D)$ to itself, since we have just seen that $\mathcal{E}$ is Lipschitz.  To see that $\mathcal{E}_{\#}$ is $\sqrt{D}e^M$-Lipschitz, fix any two $\mu,\nu\in \mathcal{P}_f(\mathbb{R}^m)$ a transport plan $\pi$ between them.  Define the ``induced diagonal transport plan'' $\tilde{\pi}:=(\mathcal{E},\mathcal{E})_{\#}\pi$ and simply note that $\tilde{\pi}$ is indeed a transport plan between $\mathcal{E}_{\#}\mu$ and $\mathcal{E}_{\#}\nu$.  We then compute
$$
\begin{aligned}
        \mathcal{W}_2^2(\mathcal{E}_{\#}\mu,\mathcal{E}_{\#}\nu)
    \le &
        \mathbb{E}_{(U_1,U_2)\sim \tilde{\pi}}[\|U_1-U_2\|^2]\\
    = &\mathbb{E}_{(V_1,V_2)\sim \pi}[\|\mathcal{E}(V_1)-\mathcal{E}(V_2)\|^2] \\
        \le &\mathbb{E}_{(V_1,V_2)\sim \pi}[De^{2M}\|V_1-V_2\|^2] \\
    = &
        De^{2M} \mathbb{E}_{(V_1,V_2)\sim \pi}[\|V_1-V_2\|^2].
\end{aligned}
$$
We complete the proof by first square-rooting both sides of the inequality and then taking the infimum over all transport plans $\pi$ between $\mu$ and $\nu$; thus
$$
\begin{aligned}
\mathcal{W}_2(\mathcal{E}_{\#}\mu,\mathcal{E}_{\#}\nu) \le \sqrt{D}e^M \inf_{\pi}
\,\mathbb{E}_{(V_1,V_2)\sim \pi}[\|V_1-V_2\|^2]^{1/2} =\sqrt{D}e^M\mathcal{W}_2(\mu,\nu).
\end{aligned}
$$
Since $\mathcal{W}_1\le \mathcal{W}_2$ for any probability measures, the second claim follows.  
\end{proof}
\begin{proof}[{Proof of Theorem~\ref{thrm:simultaneous_approximation__S}}]
By~\eqref{eq:clipped_price_process}, we have that $\mathbb{P}(S_t^x\in \cdot) = \mathcal{E}_{\#}\mathbb{P}(X_t^x\in \cdot)$.  By Lemma~\ref{lem:Gaussianity_SDE} $\mathbb{P}(X_t^x\in \cdot)$ belongs is a $D$-Dimensional Gaussian measure and therefore it belongs to $\mathcal{P}_2(\mathbb{R}^D)$.  Likewise, by construction $\hat{f}$ is also $D$-Dimensional Gaussian measure; thus, it also belongs to $\mathcal{P}_2(\mathbb{R}^D)$.  Therefore, Lemma~\ref{lem:regularity_clippedexponential} applies, allowing us to deduce that
\begin{equation}
\label{PROOF_eq:eq:clipped_price_process__LipschitzEstimate}
\begin{aligned}
        \mathcal{W}_2\big(
            \mathbb{P}(S_t^x\in \cdot)
        ,   
            \mathcal{E}_{\#}(\hat{f}(t,x))
        \big)
    \le & 
        \sqrt{D}e^M
        \,
        \mathcal{W}_2\big(
            \mathbb{P}(X_t^x\in \cdot)
        ,   
            \hat{f}(t,x)
        \big)
    .
\end{aligned}
\end{equation}
Since $\hat{f}$ is as in Theorem~\ref{thrm:Main} then, the right-hand side of~\eqref{PROOF_eq:eq:clipped_price_process__LipschitzEstimate} is less than $\sqrt{D}e^M\varepsilon$.
\end{proof}

\begin{proof}[{Proof of Theorem~\ref{thrm:Main__Pricing}}]
If $V$ is constant, then the result is clear.  Therefore, assume that $V$ is non-constant.  
Since $\mathcal{W}_1\le \mathcal{W}_2$, then Theorem~\ref{thrm:simultaneous_approximation__S} implies that, for every $(x,t)\in K\times [\delta,T]$
\begin{equation}
\label{PROOF__eq:thrm:Main__Pricing}
        \mathcal{W}_1\Big(
            \mathcal{E}_{\#}\hat{f}(x,t)
        ,
            \mathbb{Q}(\tilde{S}_t^x\in \cdot)
        \Big)
    \le 
        C_1
        \epsilon
,
\end{equation}
were the for the constant $C_1\eqdef \sqrt{D}\,e^M$.  
By the Kantorovich-Rubinstein duality (see \cite[Theorem 9.6]{gozlan2017kantorovich}) the left-hand side of~\eqref{PROOF__eq:thrm:Main__Pricing} can be rewritten as
\allowdisplaybreaks
\begin{align}
\label{PROOF_eq:Estimate}
    \sup_{
    \tilde{V}\in \operatorname{Lip}(\mathbb{R}^D,\mathbb{R};1)}\,
        \Big\vert
        \mathbb{E}_{U\sim \mathcal{E}_{\#}\hat{f}(x,t)}[\tilde{V}(U)]
        -
        \mathbb{E}_{\mathbb{Q}}\big[\tilde{V}(\tilde{S}_t^x)\big]
    \Big\vert
    =
        \mathcal{W}_1\Big(
            \hat{f}(x,t)
        ,
            \mathbb{Q}(\tilde{S}_t^x\in \cdot)
        \Big)
\end{align}
for every $x\in K$ and each $t\in [\delta,T]$; 
where $\operatorname{Lip}(\mathbb{R}^D,\mathbb{R};1)$ is the set of real-valued Lipschitz maps on $\mathbb{R}^D$ with Lipschitz norm $\|\tilde{V}\|_{Lip}\eqdef \sup_{x\in \mathbb{R}^D}\vert \tilde{V}(x)\vert + \operatorname{Lip}(\tilde{V})$ at-most $1$ \textit{(note, we the Lipschitz norm as $\infty$ if the map $\tilde{V}$ is not Lipschitz since its ``Lipschitz constant'' $\operatorname{Lip}(\tilde{V})$ would be infinite)}.  Since $V$ is non-constant then $\operatorname{Lip}(V)>0$ is positive.  Thus, $\tilde{V}\eqdef [\|V\|_{Lip}]^{-1}\cdot V$ is well-defined and~\eqref{PROOF_eq:Estimate} implies
\begin{equation}
\label{PROOF_eq:Estimate__B}
    \frac1{
        \|V\|_{Lip}
    }
    \,
        \Big\vert
        \mathbb{E}_{U\sim \mathcal{E}_{\#}\hat{f}(x,t)}[V(U)]
        -
        \mathbb{E}_{\mathbb{Q}}\big[V(\tilde{S}_t^x)\big]
        \Big\vert
    =
        \mathcal{W}_1\Big(
            \mathcal{E}_{\#}\hat{f}(x,t)
        ,
            \mathbb{Q}(\tilde{S}_t^x\in \cdot)
        \Big)
,
\end{equation}
for $(x,t) \in K\times \{T\}$.  
We conclude by multiplying~\eqref{PROOF_eq:Estimate__B} by $\|V\|_{Lip}$, setting $C\eqdef C_1\,\|V\|_{\operatorname{Lip}}$, and using the change-of-variable formula for push-forward measures.
\end{proof}
\vspace{-2em}
\bibliographystyle{splncs04}
\bibliography{Bookkeaping/2_References}
\end{document}